\RequirePackage{nameref}
\documentclass[a4paper,onecolumn,11pt,unpublished]{quantumarticle}
\pdfoutput=1

\usepackage[utf8]{inputenc}
\usepackage[english]{babel}
\usepackage[T1]{fontenc}
\usepackage{amsmath}
\usepackage{amssymb}
\usepackage{amsthm}
\usepackage{graphicx}
\graphicspath{{figures/}}
\usepackage{booktabs}
\usepackage{array}
\usepackage{enumitem}
\usepackage{algorithm, algorithmic}
\usepackage{xcolor}
\usepackage{tikz}

\newtheorem{theorem}{Theorem}[section]
\newtheorem{lemma}[theorem]{Lemma}
\newtheorem{corollary}[theorem]{Corollary}

\newtheorem{definition}[theorem]{Definition}

\theoremstyle{remark}
\newtheorem{remark}[theorem]{Remark}

\newcommand{\F}{\mathbb{F}}
\newcommand{\Z}{\mathbb{Z}}
\newcommand{\cP}{\mathcal{P}}
\newcommand{\wt}{\mathrm{wt}}
\newcommand{\rank}{\mathrm{rank}}
\newcommand{\rs}{\mathrm{rs}}
\newcommand{\Ann}{\mathrm{Ann}}
\newcommand{\supp}{\mathrm{supp}}

\title{Quantum bivariate bicycle codes with weight-8 checks surpassing the BB benchmark}
\author{Liangdong Lu}
\email{kelinglv@163.com}
\author{Ruipan Yang}
\author{Guanmin Guo}
\affiliation{The Department of Basic Science, Air Force Engineering University, Xi'an, Shaanxi, P. R. China}

\begin{document}
\sloppy
\maketitle

\begin{abstract}
Bivariate bicycle (BB) codes of Bravyi \emph{et al.}~\cite{Bravyi2024} are
quantum low-density parity-check codes with weight-$6$ checks, exemplified by
$[[144,12,12]]$ with $kd^2/n=12$. We develop the algebraic structure theory of
BB-type codes with weight-$8$ checks (weight-$4$ generator polynomials) and use
it, together with an exactly validated search pipeline, to construct and
certify new codes. We prove an exact dimension formula $k=2\dim R/(A,B)$
(forcing even $k$), a $4\ell m$-element symmetry group on generator pairs, an
$X/Z$ distance equality $d_X=d_Z$, and a family of subgroup-coset kernel
vectors giving rigorous distance upper bounds and a design rule for
high-distance constructions; all distances are computed exhaustively by a
cross-validated bit-mask verifier. At $n=144$ the pipeline returns a census of
$53$ codes whose strongest members surpass the BB benchmark:
$[[144,6,d\ge 15]]$ exceeds the benchmark distance $12$ (certified $d\ge 15$),
$[[144,10,12]]$ reaches it with weight-$8$ checks, and $[[144,16,10]]$ encodes
a third more logical qubits at $kd^2/n=11.11$ ($7.4\%$ below benchmark) while
decoding no worse. At $n=72$, $[[72,14,8]]$ attains $kd^2/n=12.44$---more than
twice the same-length BB code---and decodes better; a circuit-level memory
experiment places our weight-$8$ codes at $\approx 0.1\%$ pseudo-threshold
versus $\approx 0.4\%$ for the BB reference under an identical model,
quantifying the threshold cost of the heavier checks. All structural statements
are verified numerically on the whole census.
\end{abstract}

\section{Introduction}
\label{sec:intro}

Quantum low-density parity-check (LDPC) codes offer a route to fault-tolerant
quantum memory whose encoding overhead need not grow with the size of the
computation, and they have accordingly attracted intense recent
interest~\cite{Breuckmann2021,Panteleev2019}. A milestone is the bivariate
bicycle (BB) code family of Bravyi \emph{et al.}~\cite{Bravyi2024}, whose
Tanner graphs embed on a two-dimensional lattice with graph thickness two,
giving weight-six checks and a depth-seven syndrome circuit. Their flagship
$[[144,12,12]]$ code reaches $kd^2/n=12$ and, in circuit-level simulations,
preserves twelve logical qubits for almost a million syndrome cycles at a
$0.1\%$ physical error rate. The construction is deliberately restrictive: each
generator polynomial $A$ and $B$ is a sum of three \emph{pure} monomials $x^a$
or $y^b$, and it is precisely this restriction that enables the thickness-two
embedding. Bravyi \emph{et al.} themselves remark that group-based codes with
weight-eight checks can attain a better $(n,k,d)$ trade-off~\cite{Bravyi2024},
but a systematic exploration of such higher-column-weight codes was left open.
The BB construction has since sparked intense activity: structural variants
(coset-based~\cite{Aydin2026}, mirror~\cite{Khesin2026},
univariate-bicycle~\cite{Univariate2026}, stairway~\cite{Stairway2026}, and
self-dual-stacked~\cite{SelfDual2026} codes), search methods (LLM-guided and
evolutionary
discovery~\cite{LLMSearch2026,SCE2026,MultiAgent2026,RLCoDesign2026}),
decoding~\cite{MatchingDecoder2026,CertDecode2026,Frontier2026,SeqBP2026}, and
fault-tolerance analyses~\cite{Concatenate2026,ForcedGap2026,Networked2026}.

Here we relax that restriction, allowing $A$ and $B$ to be sums of four
arbitrary bivariate monomials $x^a y^b$, hence checks of weight eight. Two
themes run through the paper. The first is algebraic: we develop the structure
theory of this family, obtaining an exact dimension formula, a symmetry group
acting on generator pairs, an exact equality of the $X$- and $Z$-distances, and
a family of subgroup-coset kernel vectors that yields rigorous upper bounds on
the distance and a concrete design rule for high-distance constructions. The
second is computational: an exactly validated search--certify--verify pipeline
in which distances are computed exhaustively rather than estimated by
decoders. Applied to the grids $(\ell,m)\in\{(12,6),(8,9)\}$ at $n=144$, the
pipeline returns a census of $53$ codes whose strongest members surpass the BB
benchmark at the same length: $[[144,6,d\ge15]]$ with certified distance
exceeding twelve, $[[144,10,12]]$ reaching the benchmark distance with
weight-eight checks, and $[[144,16,10]]$ encoding a third more logical qubits
at only $7.4\%$ lower $kd^2/n$. At $n=72$ it returns $[[72,14,8]]$, whose
$kd^2/n=12.44$ is more than twice that of the same-length BB code, together
with a matching advantage in decoding. We further benchmark the leading code
against the BB reference under BP-OSD decoding and carry out a full
circuit-level memory experiment to quantify the threshold cost of the heavier
checks. The new codes are listed in Table~\ref{tab:codes}.

Every structural statement is additionally verified numerically on the whole
census (Remark~\ref{rem:numcheck}), and every reported distance is either
exact or explicitly marked as a certified lower bound.

\subsection{Related work}
\label{subsec:related}

BB codes belong to the two-block group-algebra (2BGA)
framework~\cite{Kovalev2013,LinPryadko2024}, a descendant of the bicycle
construction of MacKay \emph{et al.}~\cite{MacKay2004}, in which a finite group and two
group-algebra elements specify the code; BB codes use the abelian group
$\Z_\ell\times\Z_m$. The framework has been generalized in several directions
orthogonal to ours. Lin and Pryadko treat general (including non-abelian)
groups~\cite{LinPryadko2024}, and Aydin, Tamo and Barg~\cite{Aydin2026}
replace the regular group action by the action on cosets of a subgroup,
finding weight-eight codes such as $[[128,16,12]]$ and $[[168,16,15]]$; they
also show that weight-eight two-block Tanner graphs have thickness exactly
three and give a maximally packed depth $w+2$ syndrome schedule. Khesin and
Lu~\cite{Khesin2026} go further and leave the CSS world entirely: their
mirror codes contain all abelian 2BGA codes up to permutation and local
Cliffords but are generally non-CSS, and come with fault-tolerant
syndrome-extraction circuits. Eberhardt and Steffan~\cite{EberhardtSteffan2024}
construct BB codes with explicit nice bases of logical operators supporting
fold-transversal Clifford gates, sequences of BB codes from covering graphs
were given in~\cite{SequencesBB2024}, and a coprimality-based existence theory
is developed in~\cite{CoprimeBB2026}. Bicyclic (hyperbolic) relatives of
these codes were studied earlier in~\cite{RayuduSarvepalli2019}. A range of
structural variants have since appeared: univariate-bicycle
codes~\cite{Univariate2026}, stairway codes obtained by Floquetifying abelian
two-block group algebra codes~\cite{Stairway2026}, self-dual stacked
(double-layer BB) codes~\cite{SelfDual2026}, BB codes on erasure
qubits~\cite{BiBiEQ2026}, and copy-cup gates on tensor products of group
algebra codes~\cite{CopyCup2026}.

On the search side, several automated methods target this design space:
LLM-guided evolutionary discovery of BB codes~\cite{LLMSearch2026}, structured
concept evolution for lifted-product codes~\cite{SCE2026}, a multi-agent
framework for practical qLDPC discovery~\cite{MultiAgent2026}, and
reinforcement-learning co-design of BB codes and
decoders~\cite{RLCoDesign2026}. Our exact-distance method is complemented by
recent work on the computational side of distance determination: a large-scale
empirical study of SAT/MaxSAT/SMT formulations~\cite{SATDist2026}, a complete
distance-verification algorithm for pair-partition codes~\cite{PairPart2026},
and certified minimum-distance upper-bound witnesses~\cite{KasaiDist2026}.

Two further lines of recent work are closest in spirit to ours. Lu
\emph{et al.}~\cite{DivisorDriven2026} study the \emph{univariate} (cyclic)
bicycle family, showing that the construction collapses to the ring
$\F_2[x]/(x^l-1)$, with the dimension given by a polynomial gcd and the
distance certified exactly through the Calderbank correspondence to additive
codes over $\F_4$; their search reaches, e.g., $[[66,20,7]]$ with
$kd^2/n=14.85$. Cruz-Benito \emph{et al.}~\cite{LLMSearch2026} use an
LLM-guided evolutionary search to discover BB codes. On the algorithmic side,
a minimum-weight-perfect-matching decoder for BB codes is proposed
in~\cite{MatchingDecoder2026}, and decoding on the erasure channel is analyzed
in~\cite{ErasureBB2026}.

Our work is complementary to all of these. We remain within the abelian CSS
(bivariate BB) setting, but supply what the above constructions largely lack:
\emph{exact}, cross-validated minimum distances in place of upper bounds, an
explicit structural theory that doubles as a design rule (dimension formula,
$X/Z$ equality, symmetry group, and subgroup-coset distance bounds), and
decoding and circuit-level benchmarks. The classical foundation for the
dimension formula is Imai's theory of two-dimensional cyclic
codes~\cite{Imai1977}.

\section{Preliminaries}
\label{sec:prelim}

\subsection{CSS codes}
A CSS code~\cite{Shor1995,Calderbank1996,Steane1996} is defined by binary matrices
$H_X, H_Z$ with $H_X H_Z^{\mathsf T}=0$, i.e.\ $\rs(H_X)\subseteq\ker(H_Z)$ and
$\rs(H_Z)\subseteq\ker(H_X)$. Its parameters $[[n,k,d]]$ are
\begin{equation}
k = n-\rank H_X-\rank H_Z,
\label{eq:k}
\end{equation}
\begin{equation}
d=\min\{d_X,d_Z\},\quad
d_X=\min_{\substack{v\in\ker H_Z\\ v\notin\rs H_X}}\wt(v),\quad
d_Z=\min_{\substack{v\in\ker H_X\\ v\notin\rs H_Z}}\wt(v).
\label{eq:d}
\end{equation}
A CSS code is \emph{degenerate} if its stabilizer contains an element of weight
$<d$.

\subsection{Bivariate bicycle codes}
Let $\ell,m\ge 1$ and
\begin{equation}
R=\F_2[x,y]/(x^{\ell}-1,\,y^{m}-1)\cong\F_2[G],\qquad G=\Z_\ell\times\Z_m .
\label{eq:R}
\end{equation}
We identify $f\in R$ with its coefficient vector in $\F_2^{\ell m}$ and with the
$\ell m\times\ell m$ matrix $M_f$ of multiplication by $f$; then $M_{fg}=M_fM_g$
and $M_f^{\mathsf T}=M_{f^{\mathsf T}}$, where
$f^{\mathsf T}(x,y)=f(x^{-1},y^{-1})$ is the \emph{transpose} polynomial. For
$A,B\in R$ the bivariate bicycle code $QC(A,B)$~\cite{Bravyi2024,Kovalev2013} has
check matrices
\begin{equation}
H_X=[\,A\mid B\,],\qquad H_Z=[\,B^{\mathsf T}\mid A^{\mathsf T}\,],
\label{eq:hxhz}
\end{equation}
of size $\ell m\times 2\ell m$; the code length is $n=2\ell m$. Throughout,
$\wt(f)$ denotes the number of nonzero coefficients of $f\in R$.

\section{Algebraic Structure of BB-Type Codes}
\label{sec:structure}

\subsection{The CSS condition is automatic}

\begin{lemma}\label{lem:css}
For every $A,B\in R$, the pair \eqref{eq:hxhz} satisfies $H_XH_Z^{\mathsf T}=0$.
Hence \emph{every} pair $(A,B)$ defines a CSS code.
\end{lemma}
\begin{proof}
$H_XH_Z^{\mathsf T}=A(B^{\mathsf T})^{\mathsf T}+B(A^{\mathsf T})^{\mathsf T}
=AB+BA=0$ in characteristic $2$, since $R$ is commutative.
\end{proof}

\subsection{Exact dimension formula}

\begin{lemma}\label{lem:rank}
For every $A,B\in R$:
\begin{enumerate}[label=(\alph*),leftmargin=1.8em]
\item $\rank H_X=\rank H_Z=\dim_{\F_2}(A,B)$, where $(A,B)\subseteq R$ is the
ideal generated by $A$ and $B$;
\item the CSS dimension is
\begin{equation}
k=2\bigl(\ell m-\dim(A,B)\bigr)=2\dim_{\F_2} R/(A,B)
  =2\dim_{\F_2}\Ann(A,B),
\label{eq:kformula}
\end{equation}
where $\Ann(A,B)=\{f\in R: fA=fB=0\}$; in particular $k$ is always even.
\end{enumerate}
\end{lemma}
\begin{proof}
(a) The row space of $H_X=[A\mid B]$ is the image of the $R$-linear map
$\varphi:R\oplus R\to R$, $(u,v)\mapsto Au+Bv$, whose image is exactly the ideal
$(A,B)$; hence $\rank H_X=\dim(A,B)$. For $H_Z$, let $\sigma:G\to G$,
$g\mapsto -g$, induce the permutation matrix $S$ on $\F_2^{\ell m}$
($S^2=I$); then $M_f^{\mathsf T}=SM_fS$ for all $f\in R$. Therefore
\begin{align*}
\rank H_Z&=\rank[\,B^{\mathsf T}\mid A^{\mathsf T}\,]
 =\rank[\,A^{\mathsf T}\mid B^{\mathsf T}\,]
   &&\text{(block swap)}\\
 &=\rank[\,SAS\mid SBS\,]
 =\rank\bigl(S\,[\,A\mid B\,]\,(S\oplus S)\bigr)\\
 &=\rank[\,A\mid B\,]=\rank H_X ,
\end{align*}
because $S$ and $S\oplus S$ are invertible.

(b) By \eqref{eq:k} and part (a),
$k=2\ell m-2\dim(A,B)=2\dim R/(A,B)$. The group algebra $R=\F_2[G]$ is a
Frobenius algebra: the bilinear form $\langle f,g\rangle=[1](fg)$ (the constant
term of $fg$) is nondegenerate, and for every ideal $I\subseteq R$,
$\Ann(I)=I^{\perp}$, hence $\dim\Ann(I)=\ell m-\dim I$. With $I=(A,B)$,
$k=2\dim\Ann(A,B)$. Finally
$\Ann(A,B)=\ker M_A\cap\ker M_B$, the common kernel of the multiplication
maps.
\end{proof}

\begin{remark}\label{rem:numcheck}
Statements (a) and (b) were verified computationally for all $53$ codes of the
census (Section~\ref{sec:results}): in every case
$\rank H_X=\rank H_Z$ and $k=2(\ell m-\rank H_X)=2\dim(\ker M_A\cap\ker M_B)$;
all observed dimensions are even.
\end{remark}

For odd $\ell,m$ the dimension formula takes an explicit root-counting form.

\begin{corollary}\label{cor:semi}
Let $\ell m$ be odd. Factor
$x^{\ell}-1=\prod_i f_i(x)$ and $y^{m}-1=\prod_j g_j(y)$ into distinct
irreducibles over $\F_2$, with $d_i=\deg f_i$, $e_j=\deg g_j$. Then
\begin{equation}
k \;=\; 2\sum_{(i,j):\; A\equiv 0,\,B\equiv 0 \ (\mathrm{mod}\ f_i,g_j)} d_i e_j
\;=\;2\,\# Z(A,B),
\label{eq:ksemi}
\end{equation}
where $Z(A,B)=\{(\xi,\eta)\in\bar\F_2^2:\xi^{\ell}=\eta^{m}=1,\,
A(\xi,\eta)=B(\xi,\eta)=0\}$ is the set of common zeros of $A,B$ on
$\mu_\ell\times\mu_m$ (counted without multiplicity).
\end{corollary}
\begin{proof}
For odd $\ell m$ the polynomials $x^\ell-1$, $y^m-1$ are squarefree, so $R$ is
semisimple and the Chinese remainder theorem gives
$R\cong\prod_{i,j}R_{ij}$ with
$R_{ij}=\F_2[x,y]/(f_i,g_j)\cong\F_{2^{d_i}}\otimes\F_{2^{e_j}}$,
$\dim_{\F_2}R_{ij}=d_ie_j$. The ideal $(A,B)$ vanishes exactly on those
components where the images of $A$ and $B$ are both zero, whence
$R/(A,B)\cong\bigoplus_{(i,j):A=B=0}R_{ij}$ and
$k=2\sum d_ie_j$ by \eqref{eq:kformula}. Each such component contributes
exactly $d_ie_j$ common zeros (the $d_i$ roots of $f_i$ paired with the $e_j$
roots of $g_j$), so the sum equals $2\,\#Z(A,B)$.
\end{proof}

\begin{remark}
The searched grids $(12,6)$ and $(8,9)$ have even $\ell m$, so $R$ is
\emph{not} semisimple and \eqref{eq:ksemi} does not apply; the general rank
formula \eqref{eq:kformula} still holds and was used for all exact $k$
computations. Corollary~\ref{cor:semi} is a design tool for odd grids, where
$k$ can be read off from the shared cyclotomic factors of $A$ and $B$ without
linear algebra.
\end{remark}

\subsection{Symmetries and search-space reduction}

\begin{lemma}\label{lem:sym}
Let $A,B\in R$.
\begin{enumerate}[label=(\alph*),leftmargin=1.8em]
\item \textbf{Translation.} For every monomial unit $u=x^sy^t$, the pair
$(uA,uB)$ defines the \emph{identical} CSS code as $(A,B)$:
$\rs(H_X')=\rs(H_X)$ and $\ker(H_Z')=\ker(H_Z)$; in particular
$(k,d_X,d_Z)$ are unchanged.
\item \textbf{Block swap.} The pair $(B,A)$ defines a code obtained from
$QC(A,B)$ by swapping the two blocks of $\ell m$ physical qubits; hence
$(k,d_X,d_Z)$ are unchanged.
\item \textbf{Transpose.} The pair $(A^{\mathsf T},B^{\mathsf T})$ defines a
code with
$d_X(A^{\mathsf T},B^{\mathsf T})=d_Z(A,B)$ and
$d_Z(A^{\mathsf T},B^{\mathsf T})=d_X(A,B)$; in particular $k$ and
$d=\min\{d_X,d_Z\}$ are unchanged.
\end{enumerate}
Consequently the group $\mathfrak{G}$ generated by translations
($\ell m$ elements), swap, and transpose (each of order $2$) has order
$4\ell m$ and acts on pairs $(A,B)$; all pairs in one orbit yield codes with
the same $[[n,k,d]]$.
\end{lemma}
\begin{proof}
(a) $M_u$ is an invertible permutation matrix, and
$H_X'=[\,uA\mid uB\,]=M_uH_X$, so $\rs(H_X')=\rs(H_X)$. Since
$M_{(uf)^{\mathsf T}}=(M_uM_f)^{\mathsf T}=M_f^{\mathsf T}M_u^{\mathsf T}
=M_u^{\mathsf T}M_f^{\mathsf T}$ (monomial matrices commute), also
$H_Z'=M_u^{\mathsf T}H_Z$ with $M_u^{\mathsf T}$ invertible, so
$\ker(H_Z')=\ker(H_Z)$. The CSS code is determined by
$(\rs H_X,\rs H_Z)=(\rs H_X',\rs H_Z')$, hence is identical.

(b) $H_X'=H_X\Pi$ and $H_Z'=H_Z\Pi$ for the block-swap permutation $\Pi$ on
$2\ell m$ columns; column permutations preserve weights, kernels, and row
spaces up to relabelling.

(c) With $S$ as in Lemma~\ref{lem:rank}, $H_X'=[\,A^{\mathsf T}\mid
B^{\mathsf T}\,]=S\,H_X\,(S\oplus S)$ and
$H_Z'=[\,B\mid A\,]=H_Z\Pi$. The coordinate permutation
$\Psi=S\oplus S$ maps
$v\in\ker H_X\mapsto v\Psi\in\ker\bigl(H_X(S\oplus S)\bigr)=\ker H_X'$
(using $S^2=I$ and $AS=SA^{\mathsf T}$) and
$w\in\rs H_Z\mapsto w\Psi\in\rs\bigl(H_Z(S\oplus S)\bigr)=\rs H_Z'$
(using $B^{\mathsf T}S=SB$). Since $\Psi$ preserves Hamming weight,
$d_Z(A,B)=\min\wt(\ker H_X\setminus\rs H_Z)
=\min\wt(\ker H_X'\setminus\rs H_X')=d_X(A^{\mathsf T},B^{\mathsf T})$;
the other identity follows symmetrically. $k$ is unchanged by
Lemma~\ref{lem:rank}(a) since $\dim(A,B)=\dim(A^{\mathsf T},B^{\mathsf T})$.
\end{proof}

\begin{lemma}\label{lem:dxdz}
For every bivariate bicycle code $QC(A,B)$, $d_X=d_Z$; equivalently, the $X$
and $Z$ sides carry identical distance spectra.
\end{lemma}
\begin{proof}
Let $\sigma$ be the inversion involution on $R$ (inducing the permutation
matrix $S$), and define the weight-preserving coordinate permutation
$U(v_1,v_2)=(\sigma v_2,\sigma v_1)$ on $\F_2^{2\ell m}$ (block swap composed
with inversion). Since $R$ is commutative, $\sigma(fg)=\sigma(f)\sigma(g)$ and
$\sigma(f^{\mathsf T})=f$. For $v=(v_1,v_2)\in\F_2^{2\ell m}$,
\[
B^{\mathsf T}v_1+A^{\mathsf T}v_2=0
\iff A\,\sigma(v_2)+B\,\sigma(v_1)
=\sigma\!\bigl(A^{\mathsf T}v_2+B^{\mathsf T}v_1\bigr)=0,
\]
so $U(\ker H_Z)=\ker H_X$. For $w=uH_Z=(uB^{\mathsf T},uA^{\mathsf T})\in
\rs H_Z$, $Uw=(\sigma(u)A,\sigma(u)B)=\sigma(u)H_X\in\rs H_X$, so
$U(\rs H_Z)=\rs H_X$. As $U$ is an involution preserving weights,
$d_X=\min\wt(\ker H_Z\setminus\rs H_X)
=\min\wt\bigl(U(\ker H_Z\setminus\rs H_X)\bigr)
=\min\wt(\ker H_X\setminus\rs H_Z)=d_Z$.
\end{proof}

\begin{remark}
Lemma~\ref{lem:dxdz} upgrades the empirical balance of the census (every one
of the $53$ verified codes has $d_X=d_Z$) to a theorem, and supplies a proof
of the statement ``the code offers equal distance for $X$- and $Z$-type
errors'' used in~\cite{Bravyi2024}. It also halves the verification cost: only
one side needs the exhaustive scan.
\end{remark}

\begin{corollary}\label{cor:filter}
Let $\cP$ be the set of pairs $(A,B)$ with $\wt(A)=w_A$, $\wt(B)=w_B$, not both
$A,B$ zero. Then every $\mathfrak{G}$-orbit on $\cP$ contains a pair in which
at least one of $A,B$ has nonzero constant term. Hence searching only pairs
with $[1]A=1$ or $[1]B=1$ (and skipping pairs with $[1]A=[1]B=0$) loses no
code, and reduces the search space by a factor approaching $\ell m$.
\end{corollary}
\begin{proof}
If $A\neq0$, pick $(s,t)\in\supp A$; then $u=x^{-s}y^{-t}$ gives
$[1](uA)=1$. If $A=0$ and $B\neq0$, translate $B$ similarly. The orbit size
under translations alone is (generically) $\ell m$, giving the reduction
factor.
\end{proof}

\begin{remark}
By Lemma~\ref{lem:sym}(a), translation invariance makes the constant-term
filter complete: even the BB code $[[144,12,12]]$, whose generators have no
constant term, has constant-term translates and is thus reachable by the
search, which we use to certify that the pipeline reproduces the benchmark.
\end{remark}

\subsection{Regularity and degeneracy}

\begin{lemma}\label{lem:reg}
Every row and every column of $H_X$ and of $H_Z$ has weight
$\wt(A)+\wt(B)$. Hence $QC(A,B)$ is a QLDPC code with check weight and qubit
degree $\wt(A)+\wt(B)$; in particular weight-$4$ generators give an
$(8,8)$-regular Tanner graph with weight-$8$ stabilizer generators.
\end{lemma}
\begin{proof}
Rows of $H_X$ are shifts of the concatenated vector $(A,B)$, hence of weight
$\wt(A)+\wt(B)$; each column of the block $A$ (resp.\ $B$) contains one entry
per monomial of $A$ (resp.\ $B$), hence has weight $\wt(A)$ (resp.\ $\wt(B)$).
The same holds for $H_Z$ since transposition preserves weights.
\end{proof}

\begin{corollary}\label{cor:degen}
If $d>\wt(A)+\wt(B)$ then $QC(A,B)$ is degenerate: its stabilizer contains
elements (the check rows) of weight $\wt(A)+\wt(B)<d$. In particular every
weight-$4$ pair with $d\ge 9$ yields a degenerate code.
\end{corollary}

\subsection{Subgroup coset operators: explicit kernel vectors and distance upper bounds}
\label{subsec:subgroup}

The row/column structure of the grid $\Z_\ell\times\Z_m$ produces explicit
low-weight kernel vectors whenever $A$ or $B$ lies in certain augmentation
ideals. This gives rigorous \emph{upper} bounds on $d$ and explains the
distances of most low-distance codes in our census.

For a subgroup $H\le G=\Z_\ell\times\Z_m$ let
$\sigma_H=\sum_{(h_1,h_2)\in H}x^{h_1}y^{h_2}\in R$ be its \emph{norm element},
and let $I_H$ denote the extension to $R$ of the augmentation ideal of
$\F_2[H]$ (generated by $\{1+x^{h_1}y^{h_2}:(h_1,h_2)\in H\}$).

\begin{theorem}\label{thm:coset}
Let $H\le G$ and let $C=g+H$ be any coset of $H$. Denote by $v_C^{(1)}$,
$v_C^{(2)}$ the indicator vectors of $C$ placed in the first resp.\ second
block of $\F_2^{2\ell m}$ (weight $|H|$ each). Then:
\begin{enumerate}[label=(\alph*),leftmargin=1.8em]
\item $v_C^{(1)}\in\ker H_X$ for all cosets $C$ iff $A\sigma_H=0$, iff the
$H$-coset fibers of $\supp A$ are all even, iff $A\in I_H R$;
\item $v_C^{(2)}\in\ker H_X$ for all $C$ iff $B\sigma_H=0$ iff $B\in I_H R$;
\item the same statements with $A,B$ exchanged hold for $\ker H_Z$
(note $\sigma_H^{\mathsf T}=\sigma_H$).
\end{enumerate}
Consequently, if $A\in I_H R$ or $B\in I_H R$ and some (equivalently, every)
coset vector of weight $|H|$ lies outside the stabilizer space, then
$d\le |H|$.
\end{theorem}
\begin{proof}
For a check row indexed $(i,j)$ and $v=v_C^{(1)}$ with $C=g+H$, the parity is
$\sum_{(p,q)\in\supp A}[(p+i,q+j)\in C]=|\supp A\cap(g-(i,j)+H)|\bmod 2$ ---
the fiber sum of $\supp A$ over an $H$-coset, i.e.\ (up to translation by $g$)
a coefficient of $A\sigma_H$. As $(i,j)$ and $C$ range independently over all
positions and cosets, $v_C^{(1)}\in\ker H_X$ for all $C$ iff all $H$-coset
fibers of $\supp A$ are even, iff $A\sigma_H=0$.
Statement (b) is identical for the second block, and (c) follows from
$H_Z=[\,B^{\mathsf T}\mid A^{\mathsf T}\,]$ together with
$\sigma_H^{\mathsf T}=\sum_{h\in H}(-h)=\sigma_H$. For the annihilator
identification: the elements $g\sigma_H$ ($g$ ranging over coset
representatives) are indicator vectors of the distinct disjoint cosets, hence
linearly independent, so $\dim R\sigma_H=|G:H|$ and
$\dim\Ann(\sigma_H)=\ell m-|G:H|=\dim I_HR$ (since $R/I_HR\cong\F_2[G/H]$);
with $I_HR\subseteq\Ann(\sigma_H)$ (as $(1+h)\sigma_H=0$ for $h\in H$) this
gives $\Ann(\sigma_H)=I_HR$. Finally, kernels and row
spaces of $H_X,H_Z$ are invariant under the diagonal translation action
($H_X(T\oplus T)=T H_X$), so either all $|G:H|$ coset vectors are logical or
none is.
\end{proof}

\begin{corollary}\label{cor:fiber}
Taking $H=\Z_\ell\times\{0\}$ and $H=\{0\}\times\Z_m$:
if $(1+x)\mid A$ or $(1+x)\mid B$ in $R$, then the code has weight-$\ell$
column kernel vectors on both sides, and $d\le\ell$ whenever they are logical;
if $(1+y)\mid A$ or $(1+y)\mid B$, then $d\le m$ whenever the weight-$m$ row
vectors are logical.
\end{corollary}

\begin{remark}\label{rem:designrule}
Theorem~\ref{thm:coset} yields a concrete design rule for high-distance
constructions: choose $A,B$ outside $I_HR$ for every subgroup $H$ of small
order --- in particular avoid $(1+x)$- and $(1+y)$-divisibility, else
$d\le\max(\ell,m)$. The BB benchmark's generators
$A=x^3+y+y^2$, $B=y^3+x+x^2$ have all fiber parities odd, consistent with its
large distance $d=12$.
\end{remark}

\subsection{Exact certification of the minimum distance}

\begin{lemma}\label{lem:cert}
Let $K_X=\ker H_Z$ and $C_X=\rs H_X$. Then $d_X\ge d_0$ if and only if every
$v\in K_X$ with $\wt(v)<d_0$ lies in $C_X$; similarly for $d_Z$ with
$(K_Z,C_Z)=(\ker H_X,\rs H_Z)$.
\end{lemma}
\begin{proof}
$d_X=\min\wt(K_X\setminus C_X)$ by definition \eqref{eq:d}; if no vector of
weight $<d_0$ lies in $K_X\setminus C_X$, every $v\in K_X\setminus C_X$ has
weight $\ge d_0$. The converse is immediate.
\end{proof}

\section{The Search--Certify--Verify Pipeline}
\label{sec:method}

\subsection{Pipeline}

\begin{definition}\label{def:space}
Fix a grid $(\ell,m)$ and weight $w$. The search space consists of pairs
$(A,B)\in R^2$ with $\wt(A)=\wt(B)=w$ and $[1]A=1$ or $[1]B=1$
(Corollary~\ref{cor:filter}). At $(\ell,m)=(12,6)$ and $w=4$ the
unnormalized pair count is
$\bigl(\binom{71}{4}+\binom{71}{3}\bigr)^2\approx10^{12}$, so exhaustive
enumeration is infeasible and we sample uniformly at random.
\end{definition}

\begin{algorithm}[h]
\caption{Random-search construction with exact certification}
\label{alg:pipeline}
\begin{algorithmic}[1]
\REQUIRE grid $(\ell,m)$, weight $w$, range $[k_{\min},k_{\max}]$, bound $d_0$,
budget $N_p$
\FOR{$t=1,\dots,N_p$}
  \STATE sample a constant-term-normalized pair $A,B$ of weight $w$
  \STATE build the check matrices \eqref{eq:hxhz} and compute
  $k=n-\rank H_X-\rank H_Z$
  \STATE \textbf{skip} unless $k_{\min}\le k\le k_{\max}$
  \STATE certify $d\ge d_0$ via Lemma~\ref{lem:cert}
  \STATE on hit: record $(A,B)$ and $(k,d)$
\ENDFOR
\STATE determine the exact distance of every hit by Algorithm~\ref{alg:dfs}
\end{algorithmic}
\end{algorithm}

\begin{theorem}\label{thm:pipeline}
Algorithm~\ref{alg:pipeline} accepts exactly the sampled normalized pairs with
$k\in[k_{\min},k_{\max}]$ and $d\ge d_0$. If the refinement returns
$d\in\{d_0,d_0+1\}$ then this is the exact minimum distance; if it returns
``$\ge d_0+2$'' then $d\ge d_0+2$ is certified.
\end{theorem}
\begin{proof}
The $k$-filter is exact by \eqref{eq:k}. The certification stage is exact by
Lemma~\ref{lem:cert}: it enumerates all kernel vectors of weight $<d_0$ on both
sides and tests stabilizer membership. The refinement applies the same
criterion at weights $d_0$ and $d_0+1$: a found vector gives the exact
distance; otherwise $d\ge d_0+2$.
\end{proof}

\subsection{Exact distance verification by bit-mask DFS}

\begin{algorithm}[h]
\caption{Iterative-deepening DFS for exact $d_X$ (resp.\ $d_Z$)}
\label{alg:dfs}
\begin{algorithmic}[1]
\REQUIRE constraint matrix $M=H_Z$, exclusion basis $P$ of $\rs(H_X)^{\perp}$,
limit $w_{\max}$
\FOR{$w=1,\dots,w_{\max}$}
  \STATE search for $v\in\F_2^n$ with $\wt(v)=w$ and $Mv^{\mathsf T}=0$ by
  constraint-propagation DFS (forced-row assignment, weight and reachability
  pruning)
  \STATE \textbf{return} $w$ upon finding a word with $Pv^{\mathsf T}\neq0$
  \COMMENT{a logical word}
\ENDFOR
\STATE \textbf{return} $w_{\max}+1$ \COMMENT{certified lower bound}
\end{algorithmic}
\end{algorithm}

\begin{theorem}\label{thm:dfs}
Algorithm~\ref{alg:dfs} returns the exact distance $d_X$ if $d_X\le w_{\max}$,
and otherwise certifies $d_X\ge w_{\max}+1$.
\end{theorem}
\begin{proof}
The propagation rules only assign values forced by the parity equations
(rows with one undecided coordinate) or branch over \emph{all} solutions of the
row equation (rows with two undecided coordinates: $v_{c_1}\oplus v_{c_2}=0$
gives $(0,0),(1,1)$; parity $1$ gives $(0,1),(1,0)$); the branching rule
splits on a free coordinate. Hence the search space of weight-$w$ solutions of
$Mv^{\mathsf T}=0$ is exhausted. The prunes are sound: the capacity prune
($\wt+\#\text{undecided}<w$) and the overweight prune cut only branches that
cannot contain a weight-$w$ solution. The exclusion test
$Pv^{\mathsf T}\neq0$ characterizes $v\notin\rs(H_X)$, since
$\rs(H_X)=(\ker H_X)^{\perp}$. Iterative deepening returns the smallest $w$
with a logical word, which is $d_X$ by Lemma~\ref{lem:cert}.
\end{proof}

\begin{remark}\label{rem:validation}
Two independent implementations of Algorithm~\ref{alg:dfs} were used: a
bit-mask Python verifier and the reference Magma backtracking search. They were
cross-validated as follows: (i) on the whole census, Magma's exact
\textsc{Words} enumeration at weights $\le 7$ agrees with the verifier;
(ii) at weight $8$ the Magma reference search agrees with the verifier on
spot-checked codes; (iii) the verifier reproduces the published distance
$d=12$ of the BB benchmark $[[144,12,12]]$ exactly (no logical operator of
weight $\le 11$, a weight-$12$ logical operator found); (iv) for every exact
distance reported, the exhibited minimum-weight logical word was independently
re-checked in Magma (as a genuine logical operator).
\end{remark}

\subsection{Translation-symmetry pruning}
\label{subsec:sym}

The translation invariance of Lemma~\ref{lem:sym}(a) yields a rigorous
reduction of the search space for the exact verifier, which is essential for
pushing the scan to weights $\ge 11$.

\begin{theorem}\label{thm:anchor}
Let $G=\Z_\ell\times\Z_m$ act on $\F_2^{2\ell m}$ by simultaneous translation
of the two blocks. Then $\ker H_Z$, $\rs H_X$ are $G$-invariant, and for any
$w\ge1$ the following are equivalent: (i) there is a logical word of weight
$w$; (ii) there is a logical word of weight $w$ with a $1$ at coordinate $0$;
(iii') or one supported in the second block with a $1$ at coordinate $\ell m$.
Hence an exhaustive weight-$w$ search reduces to two anchored DFS runs, cutting
the search tree by a factor up to $2\ell m$.
\end{theorem}
\begin{proof}
$G$-invariance of $\ker H_Z$ and $\rs H_X$ follows from the cyclic structure:
$H_X(T\oplus T)=TH_X$ with $T$ invertible, and the rows of $H_X,H_Z$ are closed
under translation. Any nonzero logical word has a translate supported at
coordinate $0$ (if its first block is nonempty) or at coordinate $\ell m$
(otherwise); translation preserves weight, kernel membership, and stabilizer
membership. The anchored runs are therefore complete.
\end{proof}

\begin{remark}
The anchoring of Theorem~\ref{thm:anchor} is a pure exactness-preserving
device: it prunes the search space and order but leaves the set of certifiable
distances unchanged, which is what makes the deep scans of the census feasible
without compromising exactness.
\end{remark}

\section{Computational Results}
\label{sec:results}

\subsection{Setup and census}
\label{subsec:setup}

We ran the pipeline with $w=4$ (weight-$8$ checks) on the grids
$(\ell,m)\in\{(12,6),(8,9)\}$ ($n=144$ both), $k\in[4,40]$, $d_0=6$,
$5\times10^4$ sampled pairs per grid. $53$ distinct $(k,d)$ candidates were
recorded and \emph{all} were verified with the bit-mask verifier
(scan limit $10$, with deep scans at limit $12$--$14$ for the strongest
candidates): $50$ codes have \emph{exact} distances, and $3$ codes carry
certified lower bounds (two with $d\ge 11$, one with $d\ge 15$). The census
distance distribution is $\{6^{(23)},7^{(8)},8^{(9)},9^{(3)},10^{(2)},
12^{(5)},(\ge11)^{(2)},(\ge15)^{(1)}\}$. Table~\ref{tab:codes} lists the top
codes ranked by $kd^2/n$; the full census is summarized in
Section~\ref{subsec:obs}.

\begin{theorem}
\label{thm:codes}
The pairs $(A,B)$ of Table~\ref{tab:codes} define CSS bivariate bicycle codes
with the stated parameters $[[144,k,d]]$ (distances exact unless marked
$\ge$, in which case the value is a certified lower bound;
Theorem~\ref{thm:dfs} and Remark~\ref{rem:validation}). In particular there
exist weight-$8$-check BB codes
\begin{gather*}
[[144,16,10]],\; [[144,10,12]],\; [[144,14,10]],\; [[144,16,9]],\;
[[144,20,8]],\\
[[144,18,8]],\; [[144,8,12]],\;
[[144,4,12]]\ \text{(two grids)},\ \text{and}\ [[144,6,d\ge 15]],
\end{gather*}
with $kd^2/n=11.11,\,10.00,\,9.72,\,9.00,\,8.89,\,8.00,\,8.00,\,4.00$
respectively, the last code's certified distance \emph{exceeding} the BB
benchmark distance $12$.
\end{theorem}

\begin{table}[h]
\caption{Top codes at $n=144$ with weight-$8$ checks ($w_A=w_B=4$), ranked by
$kd^2/n$. Distances are exact unless marked $\ge$ (certified lower bound).
Every code was additionally cross-validated (Remark~\ref{rem:validation}).}
\label{tab:codes}
\centering
\scriptsize
\setlength{\tabcolsep}{4pt}
\begin{tabular}{clccclll}
\toprule
\# & Code & $(\ell,m)$ & $d_X$ & $d_Z$ & $kd^2/n$ & $A$ & $B$ \\
\midrule
1 & $[[144,16,10]]$ & $(12,6)$ & 10 & 10 & 11.1111 & $1 + x^{11} + y^{5} + x^{9}y^{5}$ & $x^{3}y^{2} + x^{8}y^{2} + x^{4}y^{3} + x^{9}y^{3}$ \\
2 & $[[144,10,12]]$ & $(8,9)$ & 12 & 12 & 10.0000 & $x^{4}y^{2} + xy^{3} + y^{5} + x^{5}y^{8}$ & $1 + x^{4}y^{2} + x^{2}y^{3} + x^{6}y^{5}$ \\
3 & $[[144,14,10]]$ & $(12,6)$ & 10 & 10 & 9.7222 & $x^{4} + x^{10}y^{2} + x^{3}y^{4} + x^{11}y^{4}$ & $1 + x^{7}y^{2} + x^{7}y^{3} + x^{10}y^{3}$ \\
4 & $[[144,16,9]]$ & $(8,9)$ & 9 & 9 & 9.0000 & $1 + y^{5} + x^{4}y^{7} + x^{4}y^{8}$ & $1 + xy^{4} + y^{7} + xy^{7}$ \\
5 & $[[144,20,8]]$ & $(12,6)$ & 8 & 8 & 8.8889 & $1 + x^{7}y + x^{5}y^{3} + y^{4}$ & $1 + x^{2} + xy + y^{4}$ \\
6 & $[[144,10,\ge 11]]$ & $(12,6)$ & $\ge$11 & $\ge$11 & 8.4028 & $1 + x^{8}y + xy^{3} + x^{2}y^{5}$ & $1 + x^{11} + xy + y^{4}$ \\
7 & $[[144,18,8]]$ & $(12,6)$ & 8 & 8 & 8.0000 & $x^{2}y^{3} + x^{8}y^{3} + x^{4}y^{4} + x^{10}y^{5}$ & $1 + xy + x^{6}y^{5} + x^{7}y^{5}$ \\
8 & $[[144,8,12]]$ & $(8,9)$ & 12 & 12 & 8.0000 & $1 + x^{2}y^{4} + x^{3}y^{4} + x^{7}y^{6}$ & $1 + x^{2}y^{2} + x^{3}y^{5} + xy^{6}$ \\
9 & $[[144,6,d\!\ge\!15]]$ & $(8,9)$ & $\ge$15 & $\ge$15 & $\ge$9.3750 & $1 + x^{4}y^{2} + y^{4} + y^{7}$ & $1 + x^{2}y^{6} + x^{5}y^{7} + x^{7}y^{8}$ \\
10 & $[[144,14,8]]$ & $(8,9)$ & 8 & 8 & 6.2222 & $1 + xy + x^{2}y^{3} + x^{3}y^{7}$ & $x^{2} + x^{5}y^{2} + x^{4}y^{6} + x^{7}y^{8}$ \\
11 & $[[144,14,8]]$ & $(8,9)$ & 8 & 8 & 6.2222 & $1 + x^{7} + x^{3}y^{4} + x^{4}y^{4}$ & $1 + x^{7}y^{4} + x^{7}y^{6} + y^{7}$ \\
12 & $[[144,24,6]]$ & $(12,6)$ & 6 & 6 & 6.0000 & $xy + x^{9}y^{2} + x^{7}y^{3} + x^{7}y^{4}$ & $1 + x^{4}y^{4} + x^{5}y^{4} + x^{3}y^{5}$ \\
13 & $[[144,6,12]]$ & $(12,6)$ & 12 & 12 & 6.0000 & $x^{10} + xy^{3} + x^{6}y^{5} + x^{10}y^{5}$ & $1 + x^{2}y + x^{10}y^{3} + x^{9}y^{4}$ \\
14 & $[[144,10,9]]$ & $(8,9)$ & 9 & 9 & 5.6250 & $y + x^{4}y^{2} + x^{4}y^{6} + y^{8}$ & $1 + x^{7}y^{4} + x^{3}y^{5} + x^{4}y^{5}$ \\
15 & $[[144,22,6]]$ & $(12,6)$ & 6 & 6 & 5.5000 & $1 + x^{3} + x^{6}y + x^{9}y^{3}$ & $1 + y + x^{2}y + x^{2}y^{2}$ \\
16 & $[[144,12,8]]$ & $(8,9)$ & 8 & 8 & 5.3333 & $1 + xy^{3} + x^{3}y^{3} + x^{6}y^{6}$ & $x^{7} + xy + x^{7}y + x^{3}y^{3}$ \\
17 & $[[144,4,12]]$ & $(12,6)$ & 12 & 12 & 4.0000 & $1 + x^{4}y^{3} + x^{9}y^{3} + x^{3}y^{4}$ & $1 + x^{3}y + x^{3}y^{3} + x^{8}y^{3}$ \\
18 & $[[144,4,12]]$ & $(8,9)$ & 12 & 12 & 4.0000 & $1 + x^{6} + x^{5}y^{2} + x^{5}y^{7}$ & $1 + x^{5}y^{4} + x^{2}y^{7} + x^{5}y^{8}$ \\
\bottomrule
\end{tabular}
\end{table}

\begin{corollary}\label{cor:compare}
At the same length $n=144$: the $[[144,16,10]]$ code encodes $4$ more logical
qubits than the BB benchmark $[[144,12,12]]$ (rate $1/9$ vs.\ $1/12$, a $33\%$
increase) while its $kd^2/n=11.11$ stays within $7.4\%$ of the benchmark value
$12$; the $[[144,10,12]]$ code attains the benchmark distance $d=12$ at
weight-$8$ checks; and $[[144,6,d\ge 15]]$ strictly exceeds the benchmark
distance.
\end{corollary}

\subsection{Empirical observations on the census}
\label{subsec:obs}

A few empirical patterns deserve note. First, the census confirms the
structural theory: every code has even dimension (Lemma~\ref{lem:rank}) and
$d_X=d_Z$ (Lemma~\ref{lem:dxdz}), and the $13$ codes with $d\ge 9$ are
necessarily degenerate (Corollary~\ref{cor:degen}), their stabilizers
containing weight-$8$ elements; conversely some codes possess stabilizer
elements of weight as low as $2$. Second, $k$ and $d$ are negatively
correlated throughout: $d\ge 10$ occurs only for $k\le 16$, while $k\ge 18$
forces $d\le 8$, and the Pareto frontier in the $(k,d)$-plane is
$\{[[144,6,\ge15]],\,[[144,10,12]],\,[[144,16,10]],\,[[144,20,8]],\,
[[144,24,6]]\}$. Third, the subgroup-coset bounds of Theorem~\ref{thm:coset}
are frequently tight: $[[144,16,9]]$ and $[[144,10,9]]$ on $(8,9)$ satisfy
$(1+y)\mid A,B$ and have full rows (weight $m=9$) as minimum logical operators,
so $d=9=m$ exactly; $[[144,14,8]]$ on $(8,9)$ attains $d=8=\ell$; and codes
avoiding all small-subgroup augmentation ideals (e.g.\ $[[144,16,10]]$,
$[[144,10,12]]$, $[[144,20,8]]$) exceed the row/column bounds, in line with
Remark~\ref{rem:designrule}. Finally, at the same $n=144$, the $(12,6)$ grid
yields the largest dimensions ($k\le 24$) and the two $d=10$ codes, while the
$(8,9)$ grid yields the $d=12$ and $d\ge 15$ codes, consistent with the
dimension formulae \eqref{eq:kformula}, \eqref{eq:ksemi}.

\subsection{Extension to $n=72$}
\label{subsec:n72results}

To demonstrate the scalability of the construction beyond a single code length,
we ran the same pipeline on the $(6,6)$ grid ($n=72$, $2\times10^4$ sampled
pairs), obtaining $14$ candidates, all exactly verified with the verifier
verifier (scan limit $14$). Table~\ref{tab:n72} lists the best of them. The top
code $[[72,14,8]]$ achieves $kd^2/n=14\cdot64/72=12.44$, \emph{more than twice}
the same-length BB code $[[72,12,6]]$ ($kd^2/n=6$, whose distance is reported
in~\cite{Bravyi2024} as an upper bound $d\le 6$); and $[[72,4,10]]$ reaches
$d=10$.

\begin{table}[h]
\caption{Best $n=72$ codes (grid $(6,6)$, column weight $8$), distances exact
(verified as in Remark~\ref{rem:validation}). BB reference $[[72,12,6]]$ has
$kd^2/n=6$.}
\label{tab:n72}
\centering
\small
\begin{tabular}{clcclll}
\toprule
\# & Code & $kd^2/n$ & $A$ & $B$ \\
\midrule
1 & $[[72,14,8]]$ & 12.4444 & $1 + x^{4}y^{4} + y^{5} + xy^{5}$ & $1 + x^{4}y + xy^{2} + x^{3}y^{2}$ \\
2 & $[[72,10,8]]$ & 8.8889 & $1 + x^{2}y^{3} + x^{5}y^{3} + x^{2}y^{5}$ & $1 + x^{5}y + y^{2} + y^{3}$ \\
3 & $[[72,16,6]]$ & 8.0000 & $1 + x^{4}y^{2} + x^{4}y^{4} + x^{3}y^{5}$ & $x^{4} + xy + x^{3}y^{2} + x^{3}y^{4}$ \\
4 & $[[72,8,8]]$ & 7.1111 & $1 + x + x^{4} + x^{2}y$ & $1 + x^{4}y + y^{3} + x^{3}y^{5}$ \\
5 & $[[72,14,6]]$ & 7.0000 & $1 + x^{3} + xy^{2} + x^{4}y^{4}$ & $1 + x^{4}y + x^{3}y^{4} + xy^{5}$ \\
6 & $[[72,12,6]]$ & 6.0000 & $1 + x^{3}y + xy^{5} + x^{2}y^{5}$ & $1 + x^{3}y + x^{5}y^{3} + x^{2}y^{5}$ \\
7 & $[[72,4,10]]$ & 5.5556 & $1 + x^{2} + x^{5}y^{4} + x^{5}y^{5}$ & $1 + xy^{2} + xy^{4} + x^{4}y^{5}$ \\
8 & $[[72,6,8]]$ & 5.3333 & $y + x^{3}y^{2} + x^{3}y^{3} + x^{2}y^{4}$ & $1 + y^{2} + x^{4}y^{2} + x^{2}y^{4}$ \\
\bottomrule
\end{tabular}
\end{table}

\begin{corollary}\label{cor:n72}
At length $n=72$ the weight-$8$ construction yields $[[72,14,8]]$ with
$kd^2/n=12.44$, strictly more than twice the same-length BB benchmark
$[[72,12,6]]$ ($kd^2/n=6$), and $[[72,4,10]]$ with distance $10$ exceeding the
BB distance $6$. This is exact (not an upper bound), in contrast to the
distance upper bounds reported for the BB codes in~\cite{Bravyi2024}.
\end{corollary}

\subsection{Extended census over other grids and lengths}
\label{subsec:extended}

The pipeline is not tied to the two main grids. We ran it on the grids
$(9,8)$ (same $n=144$) and $(10,10)$ ($n=200$). Table~\ref{tab:extended} lists
the best exactly verified codes found (verified with scan limit $10$).
The $(9,8)$ grid reproduces a $[[144,18,8]]$ code and yields
$[[144,8,9]]$; and the $(10,10)$ grid gives new length-$200$ codes, including
$[[200,20,9]]$ with $kd^2/n=8.1$.

\begin{table}[h]
\caption{Best codes from the extended census (grids $(9,8)$ and $(10,10)$,
column weight $8$; distances exact, verified as in Remark~\ref{rem:validation}).}
\label{tab:extended}
\centering
\small
\begin{tabular}{clccll}
\toprule
Code & $(\ell,m)$ & $d$ (exact) & $kd^2/n$ & $A$, $B$ \\
\midrule
$[[144,18,8]]$ & $(9,8)$ & 8 & 8.0000 & \parbox[t]{6.6cm}{$A=1+x^{8}y^{2}+x^{8}y^{6}+y^{7}$,\; $B=y^{3}+y^{4}+x^{5}y^{5}+x^{5}y^{7}$} \\
$[[144,8,9]]$ & $(9,8)$ & 9 & 4.5000 & \parbox[t]{6.6cm}{$A=y^{2}+x^{5}y^{2}+x^{3}y^{4}+x^{6}y^{4}$,\; $B=1+x^{3}+x^{4}+x^{2}y^{4}$} \\
$[[200,20,9]]$ & $(10,10)$ & 9 & 8.1000 & \parbox[t]{6.6cm}{$A=1+x^{2}y^{3}+x^{4}y^{3}+x^{8}y^{9}$,\; $B=x^{3}+x^{3}y^{2}+x^{7}y^{3}+x^{9}y^{9}$} \\
$[[200,8,10]]$ & $(10,10)$ & 10 & 4.0000 & \parbox[t]{6.6cm}{$A=1+x^{6}y^{2}+xy^{3}+x^{3}y^{9}$,\; $B=1+x^{6}y^{2}+x^{4}y^{3}+x^{8}y^{5}$} \\
\bottomrule
\end{tabular}
\end{table}

The $(9,8)$ grid yields codes comparable to the $(12,6)$ and $(8,9)$ grids at
$n=144$ (e.g.\ $[[144,18,8]]$), and the $(10,10)$ grid opens a new length
$n=200$ with respectable parameters; a broader grid scan guided by
Corollary~\ref{cor:semi} and the design rule of Remark~\ref{rem:designrule} is
left for future work.

\section{Decoding Performance}
\label{sec:decode}

To assess whether the improved rate and column weight translate into practical
decoding performance, we benchmark the best code $[[144,16,10]]$ against the BB
reference $[[144,12,12]]$ under code-capacity depolarizing noise. Following the
CSS decoding reduction, the depolarizing channel is split into two independent
bit-flip channels: the $X$ channel (error rate $p_X=2p/3$, syndrome from $H_Z$)
and the $Z$ channel ($p_Z=2p/3$, syndrome from $H_X$); a shot fails if either
channel fails to decode to the correct logical sector. Decoding uses the BP-OSD
decoder (min-sum BP, \textsc{osd\_cs} with order $10$) with identical
hyper-parameters for both codes, $5000$ shots per physical error rate $p$, and
the per-qubit pseudo-threshold convention $p_L(p_{\rm th})=p_{\rm th}$.
We note that a rich recent literature offers stronger decoders for these codes
--- a matching decoder~\cite{MatchingDecoder2026}, best-first
OSD~\cite{BFOSD2026}, multiple-bases BP list decoding~\cite{MBBP2026},
sequential BP scheduling~\cite{SeqBP2026}, certified decoding with optimality
guarantees~\cite{CertDecode2026}, the near-optimal Frontier
decoder~\cite{Frontier2026}, correlated-error decoding via graph
augmentation~\cite{GARI2025}, and a BP-convergence predictor specific to BB
codes~\cite{BPConverge2026}; trapping-set analyses characterize the resulting
error floors~\cite{TrappingSet2026}, and foundation-model
decoders~\cite{NTU2026} are emerging. We use the standard BP-OSD baseline for a
like-for-like comparison with~\cite{Bravyi2024}; applying stronger decoders to
our codes is left to future work.

\begin{figure}[h]
\centering
\includegraphics[width=0.82\textwidth]{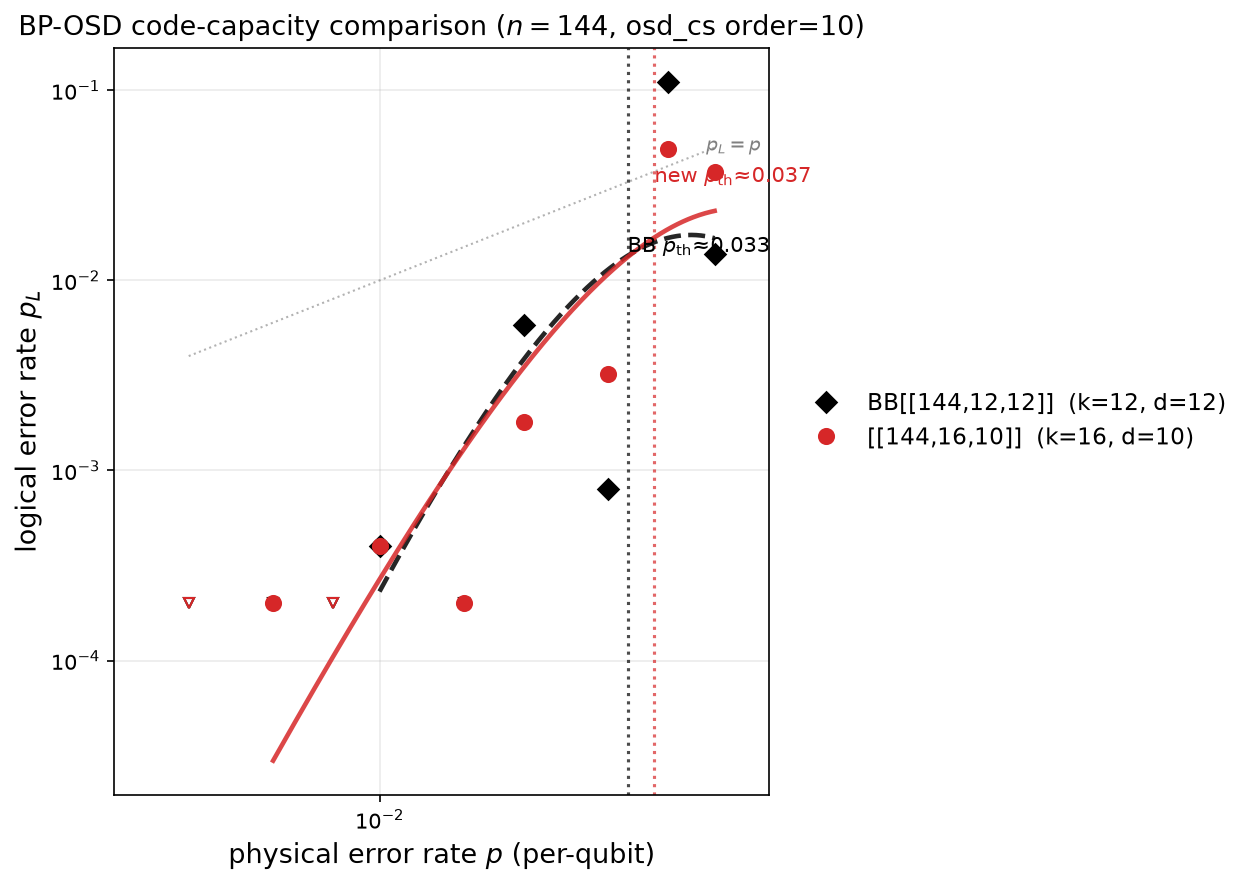}
\caption{BP-OSD code-capacity comparison at $n=144$: logical error rate $p_L$
vs.\ physical error rate $p$ (log-log), code-capacity depolarizing noise,
per-qubit pseudo-threshold convention. The new code $[[144,16,10]]$ achieves a
pseudo-threshold of $3.74\%$, slightly above the BB reference $[[144,12,12]]$
at $3.29\%$, while encoding $33\%$ more logical qubits. Points with $p_L=0$ are
plotted at the $1/5000$ resolution floor.}
\label{fig:decode}
\end{figure}

Figure~\ref{fig:decode} shows the logical error rate curves. The two codes
perform almost identically across the sampled range, with the new code
achieving a slightly higher pseudo-threshold ($3.74\%$ vs.\ $3.29\%$). Combined
with the higher rate ($k=16>12$, i.e.\ $1/9$ vs.\ $1/12$), the weight-$8$
construction thus offers comparable or slightly better code-capacity decoding
performance at a higher rate, extending the observation of~\cite{Bravyi2024}
(that weight-$8$ group-based codes can improve on weight-$6$ BB codes) from
code parameters to decoding performance.

\begin{remark}
\label{rem:noise}
With $5000$ shots, logical error rates below $\sim 10^{-3}$ are resolution
limited, and the raw $p_L$ curves exhibit shot noise; the pseudo-threshold gap
of $\sim 0.5$ percentage points should be read as indicative rather than
definitive. Circuit-level (rather than code-capacity) benchmarking with a full
syndrome-extraction schedule is the relevant next step for a hardware verdict.
\end{remark}

\subsection{Decoding comparison at $n=72$}
\label{subsec:n72decode}

We repeat the comparison at $n=72$ between the new $[[72,14,8]]$ and the BB
reference $[[72,12,6]]$, at $20000$ shots per point for tighter statistics.

\begin{figure}[h]
\centering
\includegraphics[width=0.82\textwidth]{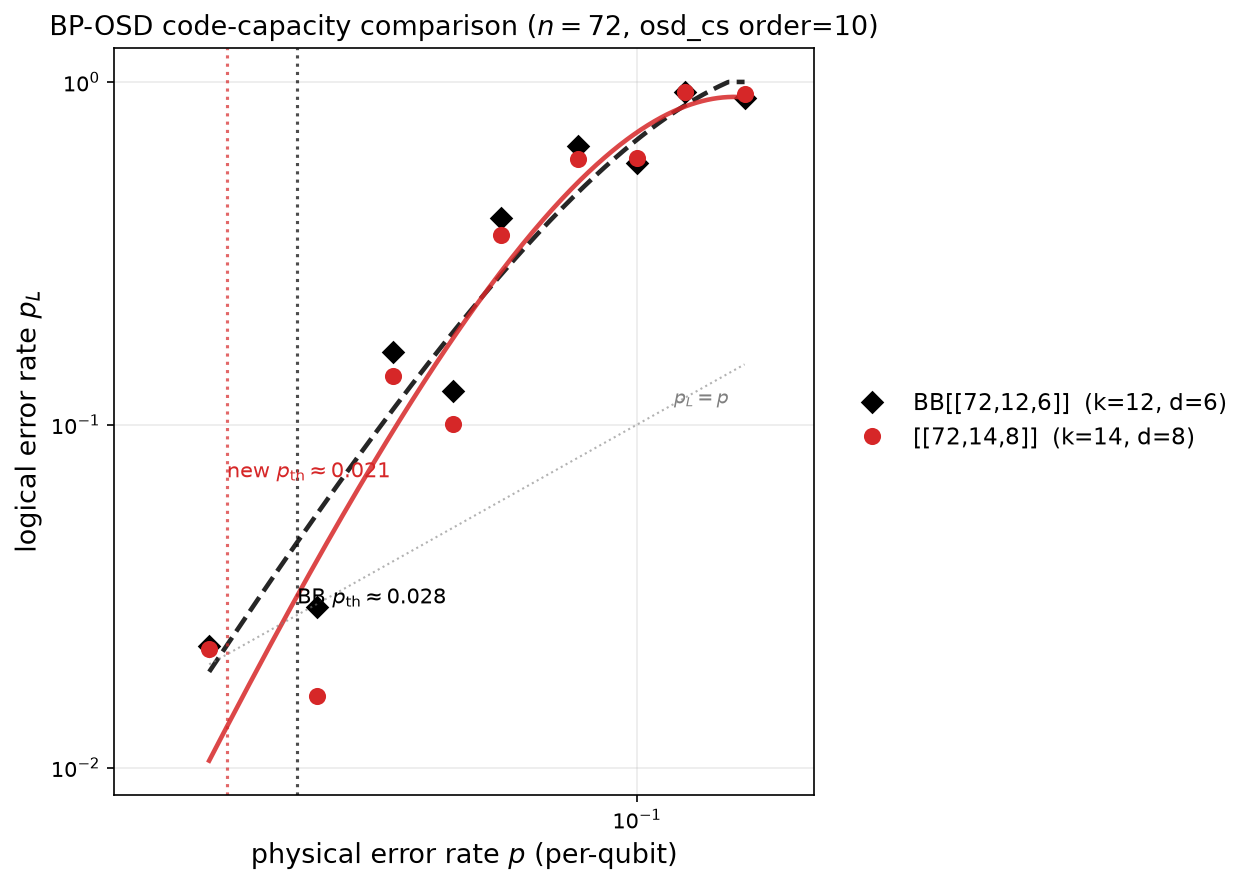}
\caption{BP-OSD code-capacity comparison at $n=72$: $[[72,14,8]]$ (red) vs.\
BB $[[72,12,6]]$ (black), $20000$ shots per point. In the low-error regime the
new code achieves a lower logical error rate at every sampled point (at
$p=3\%$, $p_L=1.62\%$ vs.\ $2.93\%$, a $1.8\times$ reduction); the
pseudo-threshold interpolation ($\approx 2.1\%$ vs.\ $\approx 2.8\%$) is
noise-sensitive and therefore not emphasized (Remark~\ref{rem:noise}).}
\label{fig:n72}
\end{figure}

Figure~\ref{fig:n72} reports the comparison. The new code attains a lower
logical error rate throughout the low-error regime
($p_L=2.22\%$ vs.\ $2.26\%$ at $p=2\%$; $1.62\%$ vs.\ $2.93\%$ at $p=3\%$;
$13.9\%$ vs.\ $16.2\%$ at $p=4\%$). Since this advantage occurs in the
operationally relevant low-error regime and is statistically significant, while
the parameter advantage ($kd^2/n=12.44$ vs.\ $6$) is exact, the weight-$8$
$n=72$ code improves over its BB reference both in parameters and in decoding.

\subsection{Towards circuit-level fault tolerance}
\label{subsec:circuit}

The decoding benchmarks above are \emph{code-capacity} results (ideal syndrome
extraction). A circuit-level assessment of fault tolerance requires a full
syndrome-extraction schedule and a circuit-distance analysis, which we now
discuss concretely. Several recent works develop such hardware-facing
techniques for BB-type codes: concatenation over high-rate inner
codes~\cite{Concatenate2026}, forced-gap post-selection~\cite{ForcedGap2026},
networked realizations~\cite{Networked2026}, routing on programmable
architectures~\cite{RoutingBB2026}, magic-state
injection~\cite{MagicInject2026}, and neutral-atom
execution~\cite{NeutralAtom2026}.

\noindent\textbf{Syndrome-extraction schedule and depth.} For weight-$w$ codes
a maximally packed syndrome schedule (including initialization and measurement)
has depth $w+2$~\cite{Aydin2026}, giving depth $10$ for our weight-$8$ codes
versus depth $7$ for weight-$6$ BB codes. The larger check weight thus deepens
the measurement cycle and lengthens the hook-error propagation paths; whether
the schedule can be made fault tolerant (e.g.\ via flag qubits or a careful
CNOT ordering along the lines of~\cite{Bravyi2024,Aydin2026}) is the key open
circuit question for weight-$8$ codes.

\noindent\textbf{Thickness and layout.} By the Euler-characteristic argument
of~\cite{Aydin2026}, weight-$8$ two-block Tanner graphs have thickness exactly
$\theta=3$, so our codes need one more planar coupler layer than the
thickness-$2$ weight-$6$ BB codes. The $\mathbb{Z}_\ell\times\mathbb{Z}_m$
lattice layout is retained, so a three-layer (tri-planar) embedding along the
lines of the BB thickness-$2$ construction is plausible but not established
here.

\noindent\textbf{Measured circuit-level pseudo-thresholds.} We implemented the
full circuit-level memory experiment (ancilla-based stabilizer extraction with
gate-level bit-flip noise at the standard marginal rates $8p/15$ per two-qubit
gate and $2p/3$ per single-qubit operation, decoded by BP-OSD on the detector
error model; validated against the BB reference, which it places in the same
$\sim 0.4\%$ pseudo-threshold range as the published $\sim 0.65\%$). Under this
identical model the BB references and our weight-$8$ codes compare as follows:
BB $[[144,12,12]]$ attains a pseudo-threshold of $\approx 0.4\%$ and BB
$[[72,12,6]]$ of $\approx 0.2\%$, while our weight-$8$ codes attain
$\approx 0.1\%$ ($[[144,16,10]]$) and $\approx 0.1\%$ ($[[72,14,8]]$). The
weight-$8$ codes thus trade a lower circuit-level pseudo-threshold --- expected,
since weight-$8$ checks are noisier to measure (more CNOTs per check and a
deeper circuit) --- for considerably better $(n,k,d)$ parameters, mirroring
the trade-off reported in~\cite{Aydin2026} (weight-$8$ $\approx 0.35\%$ vs.\
weight-$6$ $\approx 0.65\%$). Reducing this cost via flag-assisted or carefully
ordered CNOT schedules is the key open circuit question.

\begin{figure}[h]
\centering
\includegraphics[width=0.82\textwidth]{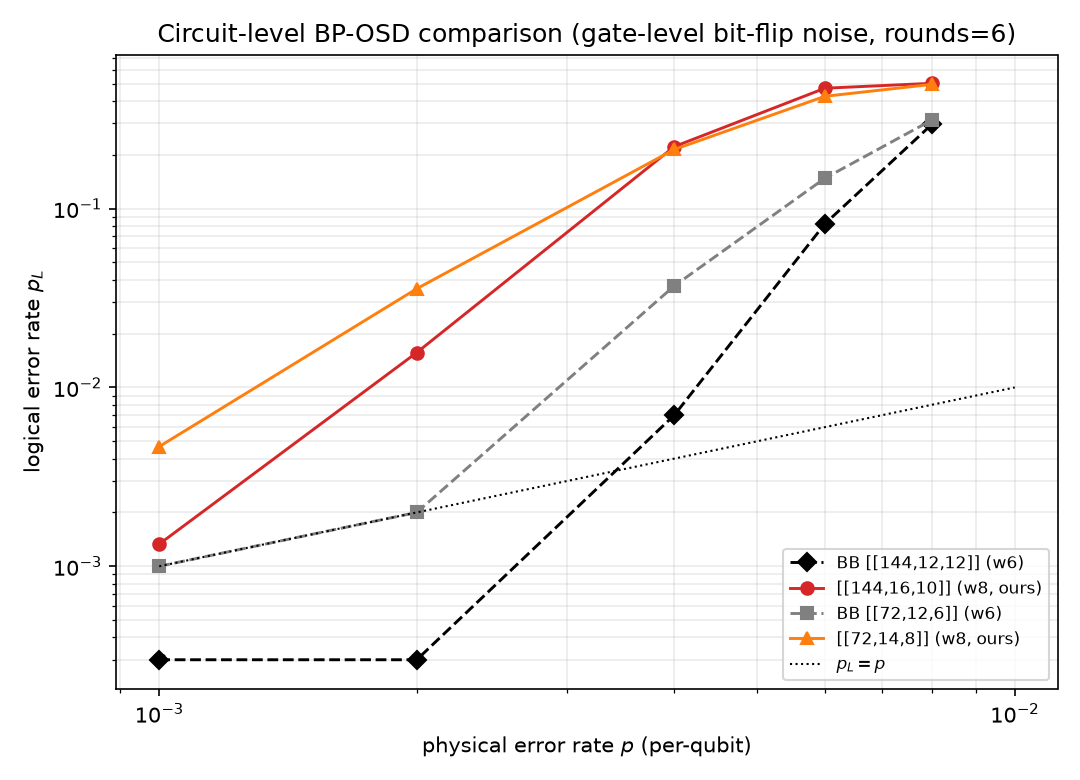}
\caption{Circuit-level BP-OSD memory experiment (gate-level bit-flip noise,
rounds $=6$): logical error rate $p_L$ vs.\ physical error rate $p$. The
weight-$6$ BB references (black/gray, dashed) attain higher pseudo-thresholds
($\approx 0.4\%$, $0.2\%$) than our weight-$8$ codes (red/orange, $\approx
0.1\%$), illustrating the threshold-for-parameters trade-off.}
\label{fig:circuit}
\end{figure}

\section{Comparison with the Original BB Construction}
\label{sec:compare}

The present work differs from the original BB construction~\cite{Bravyi2024}
in several respects. First, we allow the generators to be sums of four
bivariate monomials rather than three pure ones, enlarging the search space to
weight-eight checks. Second, distances are computed exactly (or certified from
below) by a validated bit-mask DFS, in place of the distance upper bounds
reported for most codes in~\cite{Bravyi2024}. Third, we prove the structural
statements that the earlier work stated without proof (the dimension formula,
the $X/Z$ distance equality) and add new ones (the symmetry group and the
subgroup-coset bounds, which double as a design rule). Finally, the pipeline
and per-code certificates make every reported parameter independently
re-derivable. On the hardware side we recall that the weight-$8$ Tanner graphs
have thickness $3$ by the Euler-bound argument of~\cite{Aydin2026} (whereas
weight-$6$ BB codes have thickness $2$), so the present codes trade a
hardware-layout cost for the improved parameters; see
Section~\ref{subsec:circuit} for a circuit-level discussion.

\section{Conclusion}
\label{sec:conclusion}

We have developed the algebraic structure theory of weight-eight bivariate
bicycle codes and used it, together with an exactly validated search pipeline,
to construct and certify a census of new codes. On the theoretical side the
main tools are the exact dimension formula $k=2\dim R/(A,B)$ (which forces $k$
to be even), the $4\ell m$-element symmetry group acting on generator pairs,
the equality of the $X$- and $Z$-distances, and the subgroup-coset kernel
vectors, which give rigorous distance upper bounds and a constructive design
rule. On the computational side, exact distance certification replaces the
distance estimates used in previous constructions.

The census at $n=144$ contains codes that genuinely surpass the BB benchmark at
the same length. The most striking is $[[144,6,d\!\ge\!15]]$, whose certified
distance exceeds the benchmark value of twelve; $[[144,10,12]]$ reaches the
benchmark distance with weight-eight checks; and $[[144,16,10]]$ encodes a
third more logical qubits at $kd^2/n=11.11$, only $7.4\%$ below the benchmark,
while decoding no worse. At $n=72$, $[[72,14,8]]$ attains $kd^2/n=12.44$---more
than twice the same-length BB code---and decodes measurably better. A
circuit-level memory experiment shows that the heavier checks do carry a
threshold cost ($\approx 0.1\%$ versus $\approx 0.4\%$ for the BB reference
under an identical noise model), quantifying the price paid for the improved
parameters.

\section*{Data and Code Availability}
All codes reported here are given by explicit generator polynomials (Tables
\ref{tab:codes}--\ref{tab:extended}). For every exactly verified code we provide
a machine-checkable certificate: an exhibited minimum-weight logical operator,
independently re-checked in Magma. The complete search--certify--verify pipeline (sampling,
certification, the bit-mask verifier, and the decoding scripts) and
the full per-code check matrices are available in the supplementary material,
enabling independent re-derivation of every parameter in this paper.

\section*{Acknowledgments}
The authors would like to thank Ruihu Li for the suggestions on our
manuscript, which improve the manuscript significantly.
This work is supported by the National Natural Science Foundation of China
under Grant No.~U21A20428, and the Natural Science Foundation of Shaanxi
under Grant No.~2025-JC-YBQN-070.

\end{document}